\documentclass[12pt]{article}
\newcommand{\hide}[1]{}

\usepackage{tikz}
\usepackage{etoolbox}
\usepackage{amsmath,amssymb}

\usepackage[utf8]{inputenc}
\usepackage[T1]{fontenc}

\usepackage{flafter}

\newcommand{\function}[2]{\colon #1 \rightarrow #2}

\newcommand{\setdef}[2]{\left\{ \hspace{0.5mm} #1 : \hspace{0.5mm} #2 \right\}}
\newcommand{\refeq}[1]{(\ref{eq:#1})}
 \newcommand{\Set}[1]{\big\{ #1 \big\}}
\newcommand{\barG}{\overline{G}}

\newcommand{\calH}{\ensuremath{\mathcal{H}}}
\newcommand{\calA}{\ensuremath{\mathcal{A}}}

\newcommand{\calK}{\ensuremath{\mathcal{K}}}

\newcommand{\calN}{\ensuremath{\mathcal{N}}}

\newcommand{\calR}{\ensuremath{\mathcal{R}}}

\newcommand{\tied}{\bowtie}
\newcommand{\circl}{\mathbb{C}}

\def\afterthmseparator{.}
\makeatletter
\renewcommand{\@begintheorem}[2]{\trivlist
      \item[\hskip \labelsep{\bf #1\ #2\unskip\afterthmseparator}]\em}
\renewcommand{\@opargbegintheorem}[3]{\trivlist
      \item[\hskip \labelsep{\bf #1\ #2\ (#3)\unskip\afterthmseparator}]\em}
\makeatother
\newtheorem{theorem}{Theorem}[section]
\newtheorem{lemma}[theorem]{Lemma}

\newcommand{\bull}{\mbox{$\;\;\;$\vrule height .9ex width .8ex depth -.1ex}}
\newenvironment{proof}{\par\smallbreak\noindent{\bf Proof.~}}
{\unskip\nobreak\hfill \bull \par\medbreak}

\newenvironment{proofof}[1]{\par\smallbreak\noindent{\bf Proof of~#1.~}}
{\unskip\nobreak\hfill \bull \par\medbreak}
\newcommand{\Case}[2]{\smallskip\par\emph{Case #1: #2}}
\newcounter{claim}
\renewcommand{\theclaim}{\Alph{claim}}
{\par\smallskip\par}
{$\,\triangleleft$\par\medskip\par}

\usepackage{enumitem}
\newenvironment{bfenumerate}%
{\mbox{}\begin{enumerate}[label=\arabic*.,font=\upshape\bfseries,nolistsep]}{\end{enumerate}}

\newcounter{oq}

\usetikzlibrary{calc,arrows,decorations,decorations.pathreplacing,backgrounds,shapes,fit}
\makeatletter
\pgfarrowsdeclare{interval}{interval}
{
  \pgfarrowsleftextend{+-0.5\pgflinewidth}
  \pgfarrowsrightextend{+.5\pgflinewidth}
}
{
  \pgfutil@tempdima=1pt%
  \advance\pgfutil@tempdima by1.25\pgflinewidth%
  \pgfsetdash{}{+0pt}
  \pgfsetrectcap
  \pgfpathmoveto{\pgfqpoint{0pt}{-\pgfutil@tempdima}}
  \pgfpathlineto{\pgfqpoint{0pt}{\pgfutil@tempdima}}
  \pgfusepathqstroke
}

\newif\ifdrawcapoints
\tikzset{%
  ca/.is family,
  ca/.search also={/tikz},
  ca,
  cabasestyle/.style={},
  cabase/.initial=1cm,
  castep/.initial=.2cm,
  caanglestep/.initial=1,
  calevelsx/.initial=0,
  calevelsy/.initial=0,
  calevels/.style={calevelsx=#1,calevelsy=#1},
  at calevels/.style={},
  capoints/.is if=drawcapoints,
  every carc/.style={interval},
  start/.initial={},
  end/.initial={},
  startlevel/.initial=1,
  endlevel/.initial=1,
  level/.style={startlevel=#1,endlevel=#1},
  label/.initial={},
  labelpos/.initial=0.5,
  nodepos/.style={labelpos=#1},
  labelangle/.initial={},
  nodeangle/.style={labelangle=#1},
  every label/.style={inner sep=2pt},
  startlabel/.initial={},
  startlabelsep/.initial=\CAcirclesep,
  startlabelanchoroffset/.initial=-90,
  startlabelpos/.is choice,
  startlabelpos/cw/.style={startlabelanchoroffset=+90,startlabelsep=\CAcirclesep,startlabelposauto=false},
  startlabelpos/ccw/.style={startlabelanchoroffset=-90,startlabelsep=\CAcirclesep,startlabelposauto=false},
  startlabelpos/outside/.style={startlabelanchoroffset=+180,startlabelsep=1pt+1.75\pgflinewidth,startlabelposauto=false},
  startlabelpos/inside/.style={startlabelanchoroffset=,startlabelsep=1pt+1.75\pgflinewidth,startlabelposauto=false},
  startlabelposauto/.initial=true,
  every startlabel/.style={inner sep=2pt},
  endlabel/.initial={},
  endlabelsep/.initial=\CAcirclesep,
  endlabelanchoroffset/.initial=+90,
  endlabelpos/.is choice,
  endlabelpos/cw/.style={endlabelanchoroffset=+90,endlabelsep=\CAcirclesep,endlabelposauto=false},
  endlabelpos/ccw/.style={endlabelanchoroffset=-90,endlabelsep=\CAcirclesep,endlabelposauto=false},
  endlabelpos/outside/.style={endlabelanchoroffset=+180,endlabelsep=1pt+1.75\pgflinewidth,endlabelposauto=false},
  endlabelpos/inside/.style={endlabelanchoroffset=,endlabelsep=1pt+1.75\pgflinewidth,endlabelposauto=false},
  endlabelposauto/.initial=true,
  every endlabel/.style={inner sep=2pt},
}
\newenvironment{camodel}[1][]{
  \begin{scope}[ca,#1]
    \path[ca,cabasestyle] (0,0) circle(\pgfkeysvalueof{/tikz/ca/cabase});
    \path[ca,at calevels] (0,0) ellipse
    (\pgfkeysvalueof{/tikz/ca/cabase}+\pgfkeysvalueof{/tikz/ca/calevelsx}*\pgfkeysvalueof{/tikz/ca/castep}
    and \pgfkeysvalueof{/tikz/ca/cabase}+\pgfkeysvalueof{/tikz/ca/calevelsy}*\pgfkeysvalueof{/tikz/ca/castep});
  }{
  \end{scope}
}
\newcommand{\CArc}[1]{%
  \global\def\CAinterval{}%
  \begin{scope}[interval/.append code={\global\def\CAinterval{true}},ca,#1]%
    \pgfkeys{/tikz/ca/cabase/.get=\CAbase}%
    \pgfkeys{/tikz/ca/castep/.get=\CAstep}%
    \pgfkeys{/tikz/ca/start/.get=\CAstart}%
    \pgfkeys{/tikz/ca/end/.get=\CAend}%
    \pgfkeys{/tikz/ca/startlevel/.get=\CAstartlevel}%
    \pgfkeys{/tikz/ca/endlevel/.get=\CAendlevel}%
    \pgfkeys{/tikz/ca/label/.get=\CAlabel}%
    \pgfmathparse{\CAbase+\CAstep*\CAstartlevel}%
    \dimdef\CAstartradius{\pgfmathresult pt}%
    \pgfmathparse{\CAbase+\CAstep*\CAendlevel}%
    \dimdef\CAendradius{\pgfmathresult pt}%
    \pgfmathparse{\CAstart*\pgfkeysvalueof{/tikz/ca/caanglestep}}%
    \let\CAstart\pgfmathresult
    \pgfmathparse{\CAend*\pgfkeysvalueof{/tikz/ca/caanglestep}}%
    \let\CAend\pgfmathresult
    \ifdimgreater{\CAend pt}{\CAstart pt}{%
      \edef\tmpa{\pgfkeysvalueof{/tikz/ca/startlabelposauto}}%
      \ifdefstring{\tmpa}{true}{%
        \tikzset{ca,startlabelpos=cw}%
      }{}%
      \edef\tmpa{\pgfkeysvalueof{/tikz/ca/endlabelposauto}}%
      \ifdefstring{\tmpa}{true}{%
        \tikzset{ca,endlabelpos=ccw}%
      }{}%
    }{}%
    \pgfmathparse{\CAstart\pgfkeysvalueof{/tikz/ca/startlabelanchoroffset}}%
    \let\CAstartlabelanchor\pgfmathresult
    \pgfmathparse{\CAend\pgfkeysvalueof{/tikz/ca/endlabelanchoroffset}}%
    \let\CAendlabelanchor\pgfmathresult
    \edef\CAnode{\pgfkeysvalueof{/tikz/ca/labelangle}}%
    \ifdefempty{\CAnode}{%
      \edef\CAnode{\pgfkeysvalueof{/tikz/ca/labelpos}}%
      \pgfmathparse{\CAstart+(\CAend-\CAstart)*\CAnode}%
      \let\CAnode\pgfmathresult
    }{%
      \edef\tmpa{\pgfkeysvalueof{/tikz/ca/caanglestep}}%
      \pgfmathparse{\CAnode*\tmpa}%
      \let\CAnode\pgfmathresult
    }%
    \ifdimequal{\CAstart pt}{\CAend pt}{%
      \dimdef\CAcirclesep{1pt+1.75\pgflinewidth}%
      \fill (\CAstart:\CAstartradius)
      circle (1pt+1.75\pgflinewidth);
      \ifdefempty{\CAlabel}{}{%
        \ifcsstring{tikz@auto@anchor@direction}{left}{%
          \path
          (canvas polar cs:angle=\CAstart+1,radius=\CAstartradius-1pt-1.75\pgflinewidth) --
          (canvas polar cs:angle=\CAstart-1,radius=\CAstartradius-1pt-1.75\pgflinewidth)
          node[midway,auto,ca,every label] {\CAlabel};
        }{%
          \path
          (canvas polar cs:angle=\CAstart+1,radius=\CAstartradius+1pt+1.75\pgflinewidth) --
          (canvas polar cs:angle=\CAstart-1,radius=\CAstartradius+1pt+1.75\pgflinewidth)
          node[midway,auto,ca,every label] {\CAlabel};
        }%
      }%
    }{%
      \dimdef\CAcirclesep{0pt}%
      \ifdimequal{\CAstartradius}{\CAendradius}{%
        \draw[ca,every carc] (\CAstart:\CAstartradius)
        arc[start angle=\CAstart,end angle=\CAend,radius=\CAstartradius];
      }{%
        \draw[ca,every carc,-] plot[smooth,variable=\t,domain=\CAstart:\CAend]
        ({cos(\t)*(\CAstartradius+(\t-\CAstart)/(\CAend-\CAstart)*(\CAendradius-\CAstartradius))},
        {sin(\t)*(\CAstartradius+(\t-\CAstart)/(\CAend-\CAstart)*(\CAendradius-\CAstartradius))});
        \ifdefempty{\CAinterval}{}{%
          \draw (\CAstart:\CAstartradius-1pt-1.25\pgflinewidth) --
          (\CAstart:\CAstartradius+1pt+1.25\pgflinewidth);
          \draw (\CAend:\CAendradius-1pt-1.25\pgflinewidth) --
          (\CAend:\CAendradius+1pt+1.25\pgflinewidth);
        }
      }%
      \pgfmathparse{\CAstartradius+(\CAnode-\CAstart)/(\CAend-\CAstart)*(\CAendradius-\CAstartradius)}%
      \edef\CAnoderadius{\pgfmathresult pt}%
      \ifdefempty{\CAlabel}{}{%
        \path
        (canvas polar cs:angle=\CAnode+1,radius=\CAnoderadius) --
        (canvas polar cs:angle=\CAnode-1,radius=\CAnoderadius)
        node[midway,auto,ca,every label] {\CAlabel};
      }%
      \ifdrawcapoints
        \fill (canvas polar cs:angle=\CAnode,radius=\CAnoderadius) circle (1.5pt+\pgflinewidth);
      \fi
    }%
    \path (\CAstart:\CAstartradius) node[outer sep/.expanded=\pgfkeysvalueof{/tikz/ca/startlabelsep},anchor/.expanded=\CAstartlabelanchor,ca,every startlabel] {\pgfkeysvalueof{/tikz/ca/startlabel}};
    \path (\CAend:\CAendradius) node[outer sep/.expanded=\pgfkeysvalueof{/tikz/ca/endlabelsep},anchor/.expanded=\CAendlabelanchor,ca,every endlabel] {\pgfkeysvalueof{/tikz/ca/endlabel}};
  \end{scope}%
}
\newcommand{\carc}[5][]{%
  \CArc{start={#2},end={#3},level={#4},label={#5},#1}%
}
\makeatother
\tikzset{interval/.style={interval-interval,shorten >=-.5\pgflinewidth,shorten <=-.5\pgflinewidth}}

\usepackage{hyperref}

\title{Circular-arc hypergraphs:\protect\linebreak Rigidity via Connectedness}
\author{Johannes Köbler\qquad Sebastian Kuhnert\thanks{Supported by DFG grant  KO 1053/7--1.}\qquad
Oleg Verbitsky\thanks{%
Supported by DFG grant VE 652/1--1.
This work was initiated under support by the Alexander von Humboldt Fellowship.
On leave from the Institute for Applied Problems of Mechanics and Mathematics,
Lviv, Ukraine.}\\[2.5mm]
\normalsize
Humboldt-Universität zu Berlin,
Institut für Informatik\\
\normalsize
Unter den Linden 6,
10099 Berlin, Germany}

\date{}

\begin{document} 

\maketitle

\begin{abstract}
A \emph{circular-arc hypergraph} $\calH$ is a hypergraph admitting an \emph{arc ordering}, that is, 
a circular ordering of the vertex set $V(\calH)$
such that every hyperedge is an arc of consecutive vertices.
An arc ordering is \emph{tight} if, for any two hyperedges $A$ and $B$ such that
$\emptyset\ne A\subseteq B\ne V(\calH)$, the corresponding arcs share a common endpoint.
We give sufficient conditions
for $\calH$ to have, up to reversing, a unique arc ordering and a unique tight arc ordering.
These conditions are stated in terms of connectedness properties of~$\calH$.

It is known that $G$ is a proper circular-arc graph exactly when
its closed neighborhood hypergraph $\calN[G]$ admits a tight arc ordering.
We explore connectedness properties of $\calN[G]$ and prove that, if 
$G$ is a connected, twin-free, proper circular-arc graph
with non-bipartite complement $\barG$, then $\calN[G]$ has, up to reversing, a unique arc ordering.
If $\barG$ is bipartite and connected, then $\calN[G]$ has, up to reversing, 
two tight arc orderings. As a corollary, we notice that
in both of the two cases $G$ has an essentially unique intersection representation.
The last result also follows from the work by Deng, Hell, and Huang~\cite{DengHH96} based on
a theory of local tournaments.
\end{abstract}

\section{Introduction}\label{s:intro}

\subsection{Interval and circular-arc hypergraphs}

An \emph{interval ordering} of a hypergraph $\calH$ with a finite vertex set $V=V(\calH)$ is
a linear ordering $v_1,\ldots,v_n$ of $V$ such that every hyperedge of $\calH$ is an interval of
consecutive vertices. This notion admits generalization to an \emph{arc ordering} where
$v_1,\dots,v_n$ is \emph{circularly ordered} (i.e., $v_1$ succeeds $v_n$) 
so that every hyperedge is an \emph{arc} of consecutive vertices.

An \emph{interval hypergraph} is a hypergraph
admitting an interval ordering.
Similarly, if a hypergraph admits an arc ordering, we call it \emph{circular-arc} 
(using also the shorthand \emph{CA}).
In the terminology stemming from computational genomics,
interval hypergraphs are exactly those hypergraphs
whose incidence matrix has the \emph{consecutive ones property}; e.g.,~\cite{Dom09}.
Similarly, a hypergraph is CA exactly when its
incidence matrix has the \emph{circular ones property}; e.g.,~\cite{HM03,GPZ08,OBS11}.

Our goal is to study the conditions under which interval and circular-arc hypergraphs
are \emph{rigid} in the sense that they have a unique interval or, respectively, arc
ordering. Since any interval (or arc) ordering can be changed to another interval (or arc) 
ordering by reversing, we always mean uniqueness \emph{up to reversal}.
An obvious necessary condition of the uniqueness
is that a hypergraph has no \emph{twins}, that is, no two vertices such that
every hyperedge contains either both or none of them. 

We say that two sets~$A$ and~$B$ \emph{overlap} and write $A\between B$ if
$A$ and $B$ have nonempty intersection and neither of the two sets includes the other.
To facilitate notation, we use the same character $\calH$ to denote a hypergraph
and the set of its hyperedges.
We call $\calH$ \emph{overlap-connected} if it has  no isolated vertex
(i.e., every vertex is contained in a hyperedge) and the graph
$(\calH,\between)$ is connected.
As a starting point, we refer to the following rigidity result.

\begin{theorem}[Chen and Yesha~\cite{ChenY91}]\label{thm:unique-overlap-1}
A twin-free, overlap-connected interval hypergraph has, 
up to reversal, a unique interval ordering.
\end{theorem}

If we want to extend this result to CA hypergraphs,
the overlap-connectedness obviously does not suffice.
For example, the twin-free overlap-connected hypergraph $\calH=\big\{\{a,b\},\{a,b,c\},\{b,c,d\}\big\}$
has essentially different arc orderings. 
We, therefore, need a stronger
notion of connectedness. When $A$ and $B$ are overlapping subsets of $V$
(i.e.,~$A\between B$) that additionally satisfy $A\cup B\ne V$, we say that 
$A$ and~$B$ \emph{strictly overlap} and write $A\between^* B$. 

Quilliot~\cite{Quilliot84} proves that
a CA hypergraph~$\calH$ on~$n$ vertices has a unique
arc ordering if and only if for every set $X\subset V(\calH)$ with
$1<|X|< n-1$ there exists a hyperedge $H\in\calH$ such that $H\between^* X$.
Note that this criterion does not admit efficient verification
as it involves quantification over all subsets~$X$.

We call a hypergraph~$\calH$ 
\emph{strictly overlap-connected} if it has  no isolated vertex and the graph
$(\calH,\between^*)$ is connected.
We prove the following analog of Theorem~\ref{thm:unique-overlap-1}
for CA hypergraphs.

\begin{theorem}\label{thm:unique-overlap-2}
A twin-free, strictly overlap-connected CA hypergraph has, 
up to reversal, a unique arc representation.
\end{theorem}

\subsection{Tight orderings}

Let us use notation $A\tied B$ to say that sets $A$ and $B$
have a non-empty intersection. By the standard terminology,
a hypergraph $\calH$ is \emph{connected} if it has no isolated vertex
and the graph $(\calH,\tied)$ is connected.
Note that the assumption made in Theorem~\ref{thm:unique-overlap-1}
cannot be weakened just to connectedness;
consider $\calH=\big\{\{a\},\{a,b,c\}\big\}$ as the simplest example.
Thus, if we want to weaken the assumption,
we have also to weaken the conclusion.

Call an arc ordering of a hypergraph $\calH$ \emph{tight} if,
for any two hyperedges $A$ and $B$ such that
$\emptyset\ne A\subseteq B\ne V$,
the corresponding arcs share an endpoint.
The definition of a \emph{tight interval ordering}
is similar: We require that the arcs corresponding
to hyperedges $A$ and $B$ share an endpoint whenever $\emptyset\ne A\subseteq B$
(the condition $B\ne V$ is now dropped as the complete interval $V$ has two endpoints, while the
complete arc $V$ has none).\footnote{%
The class of hypergraphs admitting a tight interval ordering
is characterized in terms of forbidden subhypergraphs in~\cite{Moore77};
such a characterization of interval hypergraphs is given in~\cite{TrotterM76}.}

For nonempty $A$ and~$B$, note that
$A\tied B$ iff $A\between B$ or $A\subseteq B$ or $A\supseteq B$.
By similarity, we define
\[
A\tied^* B\text{ iff }A\between^* B\text{ or }A\subseteq B\text{ or }A\supseteq B
\]
and say that such two nonempty sets \emph{strictly intersect}.
We call a hypergraph~$\calH$ 
\emph{strictly connected} if it has  no isolated vertex and the graph
$(\calH,\tied^*)$ is connected.
In Section~\ref{s:hgs} we establish the following result.

\begin{theorem}\label{thm:unique}
\begin{bfenumerate}
\item 
A twin-free, connected hypergraph has, up to reversal,
at most one tight interval ordering.
\item 
A twin-free, strictly connected hypergraph has, up to reversal, 
at most one tight arc ordering.
\end{bfenumerate}
\end{theorem}

\subsection{The neighborhood hypergraphs of proper interval and proper circular-arc graphs}

For a vertex $v$ of a graph $G$, the set of vertices adjacent to $v$
is denoted by $N(v)$. Furthermore, $N[v]=N(v)\cup\{v\}$.
We define the \emph{closed neighborhood hypergraph} of~$G$ 
by $\calN[G]=\{N[v]\}_{v\in V(G)}$. 

Roberts~\cite{Roberts71} discovered that $G$ is a proper interval graph 
if and only if $\calN[G]$ is an interval hypergraph.
The case of proper circular-arc (PCA) graphs is more complex.\footnote{For a
  definition of proper interval and PCA graphs, see the beginning of
  Section~\ref{s:Nhgs}.}
If $G$ is a PCA graph, then $\calN[G]$ is a CA hypergraph.
The converse is not always true.
The class of graphs with circular-arc closed neighborhood hypergraphs,
known as \emph{concave-round graphs}~\cite{Bang-JHY00},
contains PCA graphs as a proper subclass.
Taking a closer look at the relationship between PCA graphs
and CA hypergraphs, Tucker~\cite{Tucker71}\footnote{%
Tucker~\cite{Tucker71} uses an equivalent language of matrices.}
distinguishes the case when
the complement graph $\barG$ is non-bipartite and shows that then
$G$ is PCA exactly when $\calN[G]$ is~CA.

Our interest in tight orderings has the following motivation.
In fact, $G$ is a proper interval graph 
if and only if the hypergraph $\calN[G]$ has a tight interval ordering
(this follows from the Roberts theorem and Lemma~\ref{lem:geomistight} in
Section~\ref{s:Nhgs}). Moreover, $G$ is a PCA graph 
if and only if $\calN[G]$ has a tight arc ordering
(we observed this in~\cite{fsttcs} based on Lemma~\ref{lem:geomistight}
and Tucker's analysis in~\cite{Tucker71}).

Now, it is natural to consider the connectedness properties of $\calN[G]$
for proper interval and PCA graphs and derive from here
rigidity results. For proper interval graphs this issue has been
studied in the literature earlier, but we discuss also this class of graphs for
expository purposes.

We call two vertices~$u$ and~$v$ of a graph $G$ \emph{twins} if $N[u]=N[v]$. 
Note that $u$ and~$v$ are twins in the graph $G$ if
and only if they are twins in the hypergraph $\calN[G]$.
Thus, the absence of twins in $G$ is a necessary condition for rigidity of $\calN[G]$.
Another obvious necessary condition is the connectedness of $G$ (and, hence, of
$\calN[G]$).\footnote{Small graphs are an exception, as all interval orderings
  of at most two vertices are the same up to reversal, and all arc orderings of
  up to three vertices are the same up to reversal and rotation.}
By Theorem~\ref{thm:unique}.1, if a proper interval graph $G$ is
twin-free and connected, then $\calN[G]$ has a unique tight interval ordering. 
Making the same assumptions, Roberts~\cite{Roberts71} proves that
even an interval ordering of $\calN[G]$ is unique.

Suppose now that $G$ is a PCA graph. Consider first the case 
when $\barG$ is non-bipartite. In Section~\ref{s:Nhgs} we prove
that then $\calN[G]$ is strictly connected. Theorem~\ref{thm:unique}.2
applies and shows that, if $G$ is also twin-free and connected, then
$\calN[G]$ has a unique tight arc ordering.
Moreover, we prove that any arc ordering of $\calN[G]$
is tight and, hence, unique as well.

If $\barG$ is bipartite, it is convenient to switch to
the complement hypergraph $\overline{\calN[G]}=\{V(G)\setminus N[v]\}_{v\in V(G)}$.
This hypergraph is interval. Applying Theorem~\ref{thm:unique}.1
to the connected components of $\overline{\calN[G]}$,
we conclude that $\calN[G]$ has, up to reversing, exactly two tight arc orderings
provided $\barG$ is connected.

In~\cite{KoeblerKLV11} we noticed that,
if a proper interval graph $G$ is connected, then the hypergraph
$\calN[G]\setminus\{V(G)\}$ is overlap-connected.
This allows to derive Roberts' aforementioned rigidity result
from Theorem~\ref{thm:unique-overlap-1}.
In Section~\ref{s:ov-conn}, we use Theorem~\ref{thm:unique-overlap-2}
to obtain a similar result for PCA graphs:
If $G$ is an $n$-vertex connected PCA graph with non-bipartite complement,
then removal of all $(n-1)$-vertex hyperedges from $\calN[G]$
gives a strictly overlap-connected hypergraph.

\subsection{Intersection representations of graphs}

A proper interval representation $\alpha$ of a graph $G$
determines a linear ordering of $V(G)$
accordingly to the appearance of the left (or, equivalently, right)
endpoints of the intervals $\alpha(v)$, $v\in V(G)$, in the intersection model.
We call it the \emph{geometric order} associated with $\alpha$.
Similarly, a PCA representation of $G$ determines the \emph{geometric}
circular order on the vertex set. Any geometric order is a tight
interval or, respectively, arc ordering of $\calN[G]$
(see Lemma~\ref{lem:geomistight}). The rigidity results
overviewed above imply that the geometric order is unique
for twin-free, connected proper interval graphs and
twin-free, connected PCA graphs with non-bipartite complement.
In Section~\ref{s:repr} we show that this holds true also
in the case of PCA graphs with bipartite connected complement.

Let us impose reasonable restrictions on proper interval and PCA
models of graphs. Specifically, we always suppose that a model of an $n$-vertex
graph has $2n$ points and consists of intervals/arcs that never
share an endpoint. It turns out that such intersection representations
are determined by the associated geometric order uniquely up to
reflection (and rotation in the case of arcs representations).
This implies that any twin-free, connected proper interval
or PCA graph has a unique intersection representation.

The last result is implicitly contained in the work by Deng, Hell, and Huang~\cite{DengHH96}, that relies on
a theory of local tournaments~\cite{Huang95}; see the discussion in the end of Section~\ref{s:repr}.

\section{Interval and circular-arc hypergraphs}\label{s:hgs}

Let $V=\{v_1,\dots,v_n\}$.
Saying that the sequence $v_1,\dots,v_n$
is \emph{circularly ordered}, 
we mean that $V$ is endowed with
the circular successor relation~$\prec$ under which
$v_i\prec v_{i+1}$ for $i<n$ and $v_n\prec v_1$.
An ordered pair of elements $a^-,a^+\in V$ determines an \emph{arc}~$A=[a^-,a^+]$
that consists of the vertices
appearing in the directed path from~$a^-$ to~$a^+$. 
This notation will be used under the assumption that $A\ne V$, 
though we also allow the 
\emph{complete arc} $A=V$ and the \emph{empty arc} $A=\emptyset$.
By $\circl_n$ we denote the set $\{1,\ldots,n\}$
endowed with the circular order $1\prec 2\prec \ldots\prec n\prec 1$.

We here prove Theorems~\ref{thm:unique-overlap-2} and~\ref{thm:unique}.
Since the proofs are very similar, we treat in detail the most
complicated of these statements, namely part~2 of Theorem~\ref{thm:unique},
and argue that the argument also covers the other statements (including also
Theorem~\ref{thm:unique-overlap-1}).

We will use an inductive argument. To establish the rigidity of
a connected hypergraph $\calH$, we prove this property for every
connected subhypergraph $\calK\subseteq\calH$ by induction
on the number of hyperedges in $\calK$.
Since subhypergraphs of a twin-free hypergraph $\calH$
can contain twins, we need an appropriate generalization
of our rigidity concept. To this end, we switch to an equivalent language.

An \emph{arc representation of a hypergraph~$\calH$} is a hypergraph isomorphism~$\rho$
from~$\calH$ to an arc system~$\calA$ on the circle~$\circl_n$.
Note that $\calH$ has an arc representation exactly when it admits an arc ordering~$\prec$.
Indeed, if $\rho\function{V(\calH)}{\{1,\ldots,n\}}$ is an arc representation of~$\calH$,
we can define $\prec_\rho$ by $\rho^{-1}(1)\prec_\rho\rho^{-1}(2)\prec_\rho\ldots\prec_\rho\rho^{-1}(n)\prec_\rho\rho^{-1}(1)$.
Conversely, if $v_1\prec v_2\prec\ldots\prec v_n\prec v_1$ is an arc ordering of~$\calH$,
then $\rho_\prec(v_i)=i$ is an arc representation of~$\calH$. 
Furthermore, we call $\rho$ \emph{tight} if the circular order of $\circl_n$
is tight for $\calA$. Obviously, tight arc representations correspond to
tight arc orderings and vice versa.
The notions of an \emph{interval representation},
corresponding to an interval ordering, is introduced similarly.

The rotations $x\mapsto(x+s)\bmod n+1$ and the reflection $x\mapsto n+1-x$
will be considered \emph{symmetries} of the circle~$\circl_n$.
The linearly ordered segment $\{1,\ldots,n\}$ has a unique symmetry, namely the reflection.
Note that, if $\rho$~is an interval or arc representation of~$\calH$
and $\sigma$~is a symmetry of the circle or the interval, respectively,
then the composition $\sigma\circ\rho$ is an 
interval or arc representation of~$\calH$ as well.
Turning back to the equivalence between arc representations and orderings,
note that while $\rho$ determines $\prec_\rho$ uniquely, $\prec$ determines $\rho_\prec$
up to a rotation.

Now, notice that $\calH$ admits a unique, up to reversing, (tight) interval ordering
if and only if $\calH$ has a unique, up to reflection, (tight) interval representation.
Furthermore, $\calH$ admits a unique, up to reversing, (tight) arc ordering
if and only if $\calH$ has a unique, up to reflection and rotation, (tight) arc representation.

A \emph{transposition of twins} in a hypergraph $\calH$
is a one-to-one map of~$V(\calH)$ onto itself
that fixes each vertex except, possibly, a pair of twins of~$\calH$.
Any composition~$\pi$ of transpositions of twins is called a \emph{permutation of twins}.
Note that, if $\rho$~is an interval or arc representation of~$\calH$,
then the composition $\rho\circ\pi$ is, respectively, an 
interval or arc representation of~$\calH$ as well.
We call two representations~$\rho$ and~$\rho'$ equivalent \emph{up to permutation of twins}
if $\rho'=\rho\circ\pi$ for some permutation of twins $\pi$.
For twin-free hypergraphs this relation just coincides with equality of representations.

\begin{lemma}\label{lem:arcreprequi}
Arc representations~$\rho$ and~$\rho'$ are equivalent up to permutation of twins
iff $\rho(H)=\rho'(H)$ for every hyperedge~$H\in\calH$.
\end{lemma}

\begin{proof}
Suppose that $\rho'=\rho\circ\pi$ for a permutation of twins $\pi$.
Let $H$ be an arbitrary hyperedge of~$\calH$. Since $\tau(H)=H$ for any
transposition of twins $\tau$, we have $\pi(H)=H$ and, hence, $\rho(H)=\rho'(H)$.

To prove the claim in the backward direction, define
a \emph{twin-class} of a hypergraph~$\calH$ to be an inclusion-maximal subset~$S$ of~$V(\calH)$ such that
each hyperedge~$H\in\calH$ contains either all of~$S$ or none of it.
Thus, two vertices of~$\calH$ are twins iff they are in the same twin-class.
It follows that $\pi$~is a permutation of twins exactly when 
$\pi(S)=S$ for every twin-class $S$ of~$\calH$.

Suppose that $\rho(H)=\rho'(H)$ for every hyperedge~$H\in\calH$.
Since a twin-class is a Boolean combination of hyperedges,
we have $\rho(S)=\rho'(S)$ for every twin-class $S$ of~$\calH$ and,
therefore, $\rho'\circ\rho^{-1}(S)=S$ for all twin-classes.
It follows that $\rho'\circ\rho^{-1}$ is a permutations of twins,
hence $\rho$ and~$\rho'$ are equivalent up to permutation of twins.
\end{proof}

Let $\rho\function{V(\calH)}{\{1,\ldots,n\}}$ and $\rho'\function{V(\calH)}{\{1,\ldots,n\}}$
be arc or interval representations of~$\calH$. We call $\rho$ and~$\rho'$
equivalent \emph{up to symmetry and permutation of twins} if 
$\rho'=\sigma\circ\rho\circ\pi$ for some symmetry $\sigma$ and
permutation of twins $\pi$. This is an equivalence relation between
representations because both symmetries and permutations of twins form groups.

The following lemma translates Theorems~\ref{thm:unique-overlap-1},~\ref{thm:unique-overlap-2}, and~\ref{thm:unique}
into the language of interval/arc representations and generalizes these results
to hypergraphs with twins. Theorems~\ref{thm:unique-overlap-1},~\ref{thm:unique-overlap-2}, and~\ref{thm:unique}
follow from here immediately by the equivalence between representations and orderings.

\begin{lemma}\label{lem:unique}
\begin{bfenumerate}
\item
An overlap-connected interval hypergraph~$\calH$~has, 
up to symmetry (i.e.\ reflection) and permutation of twins, 
a unique interval representation.
\item
A strictly overlap-connected CA hypergraph~$\calH$~has, 
up to symmetry (i.e., reflection and rotation) and permutation of twins, 
a unique arc representation.
\item 
A connected hypergraph has, up to symmetry and permutation of twins,
at most one tight interval representation.
\item 
A strictly connected hypergraph has, up to symmetry
and permutation of twins, at most one tight arc representation.
\end{bfenumerate}
\end{lemma}

\begin{proof}
We will prove part~4. The proof of the interval variant (part~3) is virtually the same
and even somewhat simpler as not all arc configurations can occur on the line (like one of the
dashed configurations in Fig.~\ref{fig:proof:unique2} that closes the circle).
Parts 1 and 2 are actually covered by the argument as well;
they correspond exactly to the case shown in Fig.~\ref{fig:proof:unique2}
as all the other cases involve inclusions.

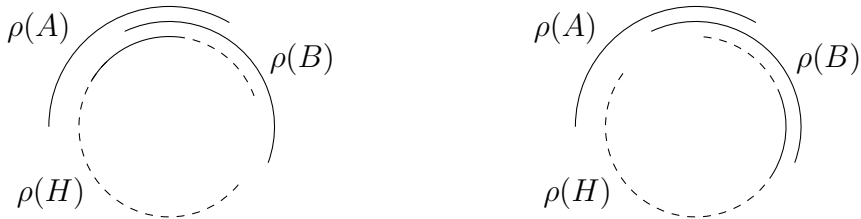
\begin{figure}
  \centering
  \begin{tikzpicture}
    \begin{camodel}[every carc/.style={}]
      \carc[labelpos=.33]{180}{60}{3}{$\rho(A)$}
      \carc[labelpos=.66]{115}{-20}{2}{$\rho(B)$}
      \carc{150}{85}{1}{}
      \carc[dashed]{85}{20}{1}{}
      \carc[dashed,labelangle=-150]{-40}{-230}{1}{$\rho(H)$}
    \end{camodel}
    \begin{camodel}[xshift=7cm,every carc/.style={}]
      \carc[labelpos=.33]{180}{60}{3}{$\rho(A)$}
      \carc[labelpos=.66]{115}{-20}{2}{$\rho(B)$}
      \carc[dashed]{85}{20}{1}{}
      \carc{20}{-30}{1}{}
      \carc[dashed,labelangle=-150]{-30}{-220}{1}{$\rho(H)$}
    \end{camodel}
  \end{tikzpicture}
  \caption{Proof of Lemma~\protect\ref{lem:unique}.4, case 1: $A\between^*B$ and
$\rho(B)$~intersects~$\rho(A)$ clockwise; $B\between^*H$.
On the left side:
 $\rho(H)$~intersects~$\rho(B)$ counter-clockwise.
On the right side:
 $\rho(H)$~intersects~$\rho(B)$ clockwise.}\label{fig:proof:unique2}
\end{figure}

Let us first explain the strategy of the proof.
A set of hyperedges $\calK\subseteq\calH$ will be regarded as
a \emph{spanning subhypergraph} of $\calH$, that is, $V(\calK)=V(\calH)$
(note that $\calK$ can have isolated vertices).
Given a hypergraph~$\calH$ with $n$ vertices,
we will prove that any strictly connected subhypergraph $\calK\subseteq\calH$ with $k$ hyperedges
has, up to symmetry and permutation of twins, a unique representation on the circle of size~$n$. 
This will be done by induction on~$k$. The base cases are $k=1,2$.
In order to make the inductive step, it suffices to show that, whenever $k\ge2$,
any representation~$\rho$ of~$\calK$ has, up to permutation of twins, a unique
extension to a representation of $\calK\cup\{H\}$,
for any $H\in\calH\setminus\calK$ such that $\calK\cup\{H\}$ is strictly connected.
By Lemma~\ref{lem:arcreprequi} this actually means to show that the whole arc~$\rho(H)$
(though not necessarily each point~$\rho(v)$ for $v\in H$) is uniquely determined.
Moreover, it suffices to do this job in the case of $k=2$.
The reason is that $\calK$ always contains two hyperedges~$A$ and~$B$ such that
the sequence $A,B,H$ forms a strictly connected path.

Before going into detail, we introduce some terminology.
Consider two arcs $[a^-,a^+]$ and $[b^-,b^+]$ and suppose that
$[b^-,b^+]\between^*[a^-,a^+]$ or $[b^-,b^+]\subset[a^-,a^+]$.
We will say that $[b^-,b^+]$ \emph{intersects} $[a^-,a^+]$ \emph{clockwise}
if $a^+\in[b^-,b^+]$ and \emph{counter-clockwise} if $a^-\in[b^-,b^+]$.

All possible positions of a single hyperedge~$A$ on the circle are congruent by rotation.
All possible positions of two strictly overlapping hyperedges~$A$ and~$B$
are congruent by rotation and reflection because the intersection of the corresponding
arcs~$\rho(A)$ and~$\rho(B)$ is always an arc of length $|A\cap B|$. 
If $A$ and~$B$ are comparable under
inclusion, recall that we only consider tight representations.

For the inductive step, consider three hyperedges~$A$, $B$, and~$H$
such that $A\tied^*B$ and $B\tied^*H$.
We have to show that the arc~$\rho(H)$ is completely determined by
the arcs~$\rho(A)$ and~$\rho(B)$.
The relation $B\tied^*H$ fixes the length of the intersection $\rho(H)\cap\rho(B)$ 
and, hence, leaves for~$\rho(H)$ exactly two possibilities depending on whether
this intersection is clockwise or counter-clockwise.

We split our analysis into three cases.

\Case1{$A\between^*B$.}
Without loss of generality, we suppose that $\rho(B)$ intersects~$\rho(A)$
clockwise; the case of counter-clockwise intersection is symmetric. 
Consider first the subcase in which $B\between^*H$.
Looking at the possible configurations
for the arc system $\{\rho(A),\rho(B),\rho(H)\}$, 
all shown in Fig.~\ref{fig:proof:unique2},
we see that $\rho(H)$ intersects~$\rho(B)$ counter-clockwise
exactly if the sets $A\setminus B$ and $H\setminus B$ are comparable under inclusion, i.e.,
\begin{equation}
  \label{eq:ccw}
A\setminus B\subseteq H\setminus B\text{\ \ or\ \ }H\setminus B\subseteq A\setminus B.
\end{equation}

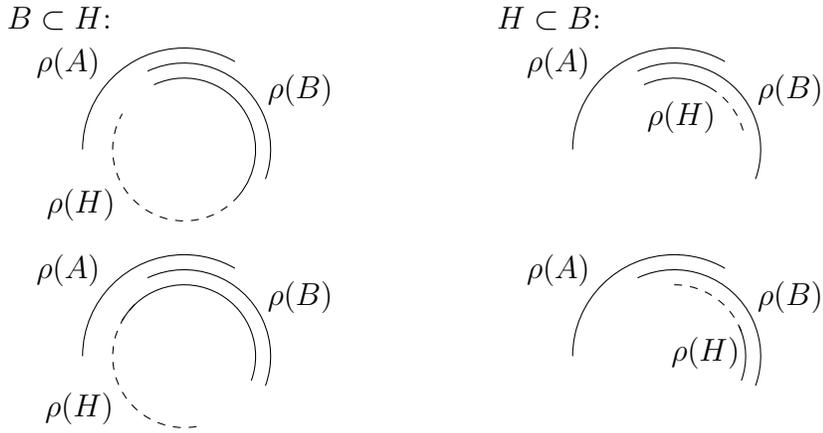
\begin{figure}
  \centering
  \begin{tikzpicture}[baseline=0cm]
    \begin{camodel}[every carc/.style={},cabase=.75cm]
      \node[anchor=west] at (-2.5cm,1.75cm) {$B\subset H$:};
      \carc[nodepos=.33]{180}{60}{3}{$\rho(A)$}
      \carc[nodepos=.66]{115}{-20}{2}{$\rho(B)$}
      \carc{115}{-40}{1}{}
      \carc[dashed,nodepos=.66]{-40}{-210}{1}{$\rho(H)$}
    \end{camodel}
    \begin{camodel}[yshift=-2.75cm,every carc/.style={},cabase=.75cm]
      \carc[nodepos=.33]{180}{60}{3}{$\rho(A)$}
      \carc[nodepos=.66]{115}{-20}{2}{$\rho(B)$}
      \carc{150}{-20}{1}{}
      \carc[dashed]{-80}{-230}{1}{$\rho(H)$}
    \end{camodel}
  \end{tikzpicture}\hfil
  \begin{tikzpicture}[baseline=0cm]
    \begin{camodel}[every carc/.style={},cabase=.75cm]
      \node[anchor=west] at (-2.5cm,1.75cm) {$H\subset B$:};
      \carc[nodepos=.33]{180}{60}{3}{$\rho(A)$}
      \carc[nodepos=.66]{115}{-20}{2}{$\rho(B)$}
      \carc{115}{60}{1}{}
      \carc[dashed,swap,nodepos=.25]{60}{15}{1}{$\rho(H)$}
    \end{camodel}
    \begin{camodel}[yshift=-2.75cm,every carc/.style={},cabase=.75cm]
      \carc[nodepos=.33]{180}{60}{3}{$\rho(A)$}
      \carc[nodepos=.66]{115}{-20}{2}{$\rho(B)$}
      \carc[swap]{25}{-20}{1}{$\rho(H)$}
      \carc[dashed]{90}{20}{1}{}
    \end{camodel}
  \end{tikzpicture}
  \caption{Proof of Lemma~\protect\ref{lem:unique}.4, case 1: $A\between^*B$ and
    $\rho(B)$ intersects~$\rho(A)$ clockwise;
$B$ and~$H$ are comparable under inclusion.}\label{fig:proof:unique4a}
\end{figure}

For the remaining subcases, when $B$ and~$H$ are comparable under inclusion,
all possible configurations for the arc system $\{\rho(A),\rho(B),\rho(H)\}$
are shown in Fig.~\ref{fig:proof:unique4a}. If $B\subset H$, 
we see that $\rho(B)$ intersects~$\rho(H)$ clockwise
exactly under the condition~\refeq{ccw}.
If $H\subset B$, then $\rho(H)$ intersects~$\rho(B)$ counter-clockwise iff
$A\cap B\subseteq H$.

\begin{figure}
  \centering
  \begin{tikzpicture}[baseline=0cm]
    \begin{camodel}[every carc/.style={},cabase=.75cm]
      \node[anchor=west] at (-2.25cm,1.75cm) {$B\between^*H$:};
      \carc[nodepos=.25]{180}{0}{2}{$\rho(A)$}
      \carc{90}{0}{3}{$\rho(B)$}
      \carc{135}{45}{1}{}
      \carc[dashed]{280}{135}{1}{$\rho(H)$}
    \end{camodel}
    \begin{camodel}[yshift=-2.75cm,every carc/.style={},cabase=.75cm]
      \carc[nodepos=.25]{180}{0}{2}{$\rho(A)$}
      \carc{90}{0}{3}{$\rho(B)$}
      \carc[labelangle=-30]{45}{-45}{1}{$\rho(H)$}
      \carc[dashed]{-45}{-215}{1}{}
    \end{camodel}
  \end{tikzpicture}\hfill
  \begin{tikzpicture}[baseline=0cm]
    \begin{camodel}[xshift=5cm,every carc/.style={},cabase=.75cm]
      \node[anchor=west] at (-2.25cm,1.75cm) {$B\subset H$:};
      \carc[nodepos=.25]{180}{0}{2}{$\rho(A)$}
      \carc{90}{0}{3}{$\rho(B)$}
      \carc{135}{0}{1}{}
      \carc[dashed]{280}{135}{1}{$\rho(H)$}
    \end{camodel}
    \begin{camodel}[xshift=5cm,yshift=-2.75cm,every carc/.style={},cabase=.75cm]
      \carc[nodepos=.25]{180}{0}{2}{$\rho(A)$}
      \carc{90}{0}{3}{$\rho(B)$}
      \carc[labelangle=-30]{90}{-45}{1}{$\rho(H)$}
      \carc[dashed]{-45}{-215}{1}{}
    \end{camodel}
  \end{tikzpicture}\hfill
  \begin{tikzpicture}[baseline=0cm]
    \begin{camodel}[xshift=10cm,every carc/.style={},cabase=.75cm]
      \node[anchor=west] at (-2.25cm,1.75cm) {$H\subset B$:};
      \carc[nodepos=.25]{180}{0}{2}{$\rho(A)$}
      \carc{90}{0}{3}{$\rho(B)$}
      \carc[swap]{55}{0}{1}{$\rho(H)$}
    \end{camodel}
  \end{tikzpicture}
  \caption{Proof of Lemma~\protect\ref{lem:unique}.4, case 2: $B\subset A$ and
    $\rho(B)$ intersects~$\rho(A)$ clockwise.}\label{fig:proof:unique4b}
\end{figure}
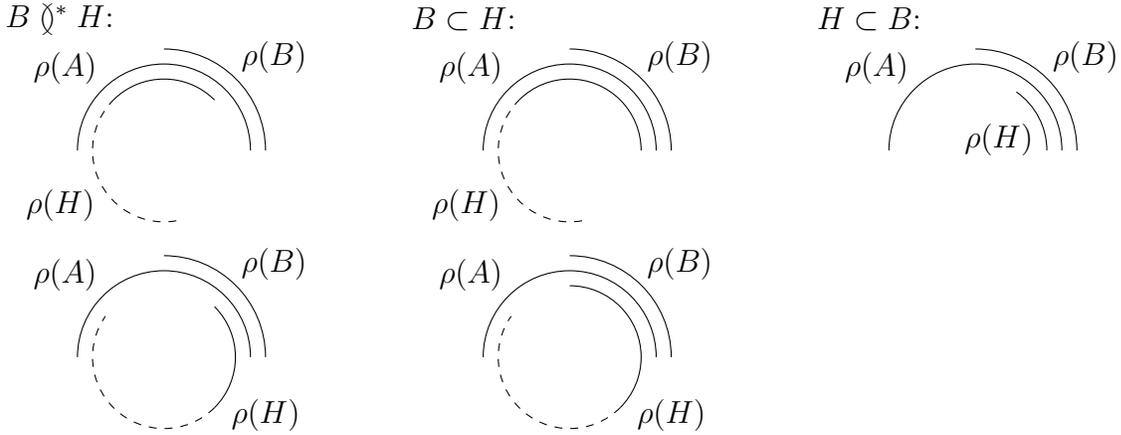

\Case2{$A\supset B$.}
Without loss of generality, we suppose that $\rho(B)$ intersects~$\rho(A)$
clockwise; see Fig.~\ref{fig:proof:unique4b}.
If $B\between^*H$, then $\rho(H)$ intersects~$\rho(B)$ counter-clockwise
exactly when the familiar condition~\refeq{ccw} holds true.
If $B\subset H$, then $\rho(B)$ intersects~$\rho(H)$ clockwise
exactly under the same condition. Thus, these two subcases do not differ much
from the corresponding subcases of Case 1.
If $H\subset B$, then $\rho(H)$~is forced to intersect~$\rho(B)$ clockwise
by the condition that $\{\rho(A),\rho(B),\rho(H)\}$ is a tight arc system.

\begin{figure}
  \centering
  \begin{tikzpicture}[baseline=0pt]
    \begin{camodel}[every carc/.style={},cabase=.75cm]
      \node[anchor=west] at (-2.25cm,1.75cm) {$B\between^*H$:};
      \carc[nodepos=.25]{180}{0}{2}{$\rho(B)$}
      \carc{90}{0}{3}{$\rho(A)$}
      \carc[labelpos=.4]{280}{135}{1}{$\rho(H)$}
      \carc[dashed]{135}{45}{1}{}
    \end{camodel}
    \begin{camodel}[yshift=-2.75cm,every carc/.style={},cabase=.75cm]
      \carc[nodepos=.25]{180}{0}{2}{$\rho(B)$}
      \carc{90}{0}{3}{$\rho(A)$}
      \carc{50}{-100}{1}{$\rho(H)$}
      \carc[dashed]{90}{50}{1}{}
    \end{camodel}
  \end{tikzpicture}\hfill
  \begin{tikzpicture}[baseline=0pt]
    \begin{camodel}[every carc/.style={},cabase=.75cm]
      \node[anchor=west] at (-2.25cm,1.75cm) {$H\subset B$:};
      \carc[nodepos=.25]{180}{0}{2}{$\rho(B)$}
      \carc{90}{0}{3}{$\rho(A)$}
      \carc[swap]{180}{130}{1}{$\rho(H)$}
      \carc[dashed]{130}{45}{1}{}
    \end{camodel}
    \begin{camodel}[yshift=-2.75cm,every carc/.style={},cabase=.75cm]
      \carc[nodepos=.25]{180}{0}{2}{$\rho(B)$}
      \carc{90}{0}{3}{$\rho(A)$}
      \carc[swap]{50}{0}{1}{$\rho(H)$}
      \carc[dashed]{135}{50}{1}{}
    \end{camodel}
  \end{tikzpicture}\hfill
  \begin{tikzpicture}[baseline=0pt]
    \begin{camodel}[every carc/.style={},cabase=.75cm]
      \node[anchor=west] at (-2.25cm,1.75cm) {$B\subset H$:};
      \carc[nodepos=.25]{180}{0}{2}{$\rho(B)$}
      \carc{90}{0}{3}{$\rho(A)$}
      \carc[nodepos=.25]{270}{0}{1}{$\rho(H)$}
    \end{camodel}
  \end{tikzpicture}
  \caption{Proof of Lemma~\protect\ref{lem:unique}.4, case 3: $A\subset B$ and
    $\rho(A)$ intersects~$\rho(B)$ clockwise.}\label{fig:proof:unique4c}
\end{figure}

\Case3{$A\subset B$.}
Without loss of generality, we suppose that $\rho(A)$ intersects~$\rho(B)$
clockwise; see Fig.~\ref{fig:proof:unique4c}.
If $B\between^*H$, then $\rho(H)$ intersects~$\rho(B)$ clockwise
iff $H\cap B\subseteq A$.
If $H\subset B$, then $\rho(H)$ intersects~$\rho(B)$ clockwise
iff the sets~$H$ and~$A$ are comparable under inclusion.
Finally, if $B\subset H$, then $\rho(B)$~is forced to intersect~$\rho(H)$ clockwise
by the tightness condition.
\end{proof}

\section{The neighborhood hypergraphs}
\label{s:Nhgs}

Let $\calA$ be a family of arcs on the circle $\circl_m$.
A bijection $\alpha\function{V(G)}{\calA}$ is an \emph{arc representation
of a graph $G$} if two vertices $u$ and $v$ are adjacent in $G$ exactly when
the arcs $\alpha(u)$ and $\alpha(v)$ intersect.
A representation $\alpha$ is \emph{proper} if $\alpha(u)\subseteq\alpha(v)$
for no two vertices $u$ and $v$. Restriction to \emph{intervals}
in the linearly ordered set $\{1,\ldots,m\}$ gives the notion of
a \emph{proper interval representation} of $G$.
Graphs having such intersection representations
are known as \emph{proper interval} and \emph{proper circular-arc (PCA) graphs}.

The aim of this section is to prove the rigidity properties for the closed
neighborhood hypergraphs of PCA graphs that will be given in
Theorem~\ref{thm:Nrigid}.

Recall that a proper arc (resp.\ interval)
representation $\alpha\function{V(G)}{\calA}$ of a graph~$G$
determines the circular (resp.\ linear) geometric order~$\prec_\alpha$ on the vertex set $V(G)$
accordingly to the appearance of the left (or, equivalently, right)
endpoints of the arcs $\alpha(v)$, $v\in V(G)$, in the circle $\circl_m$. 
The following lemma implies that the closed neighborhood hypergraph of any PCA
graph admits a tight arc ordering.

\begin{lemma}[see~\cite{fsttcs}]\label{lem:geomistight}
  The geometric order~$\prec_\alpha$ on~$V(G)$ associated with a proper
  arc (resp.\ interval) representation~$\alpha$ of a graph~$G$ is a tight arc
(resp.\ interval) ordering of the hypergraph~$\calN[G]$.
\end{lemma}

For the remainder of this section, we will consider two subclasses of PCA
graphs, namely those with bipartite and those with non-bipartite complement.
 The \emph{complement of a graph~$G$} is the graph~$\barG$ with
 $V(\barG)=V(G)$
 such that two vertices are adjacent in~$\barG$
 if and only if they are not adjacent in~$G$. 

In what follows, we will repeatedly need the following property of
non-co-bipartite PCA graphs.
A vertex~$u$ is \emph{universal} if $N[u]=V(G)$.

\begin{lemma}\label{lem:nouniv}
  A PCA graph $G$ with non-bipartite complement
contains no universal vertex.
\end{lemma}

\begin{proof}
Let~$\alpha\colon V(G)\to\calA$ be a
    proper arc representation of~$G$ and 
denote~$\alpha(u)=[a^-,a^+]$. 
Notice now that $N[u]$ is
    covered by two cliques $\setdef{x\in
      N[u]}{a^-\in\alpha(x)}$ and $\setdef{x\in
      N(u)}{a^+\in\alpha(x)}$,
which excludes $N[u]=V(G)$ as $\barG$ is not bipartite.  
\end{proof}

If $\prec$ is an arc ordering of the hypergraph $\calN[G]$ for
a non-co-bipartite PCA graph $G$, the closed neighborhood $N[u]$
of a vertex $u$ is an arc on the directed cycle $(V(G),\prec)$.
Lemma~\ref{lem:nouniv} allows us to use notation $N[u]=[u^-,u^+]$,
since the left endpoint $u^-$ and the right endpoint $u^+$
are uniquely determined.

In the following three lemmas, we establish several facts about arc orderings of
non-co-bipartite PCA graphs.

\begin{lemma}\label{lem:twocliques}
Let $G$ be a PCA graph with non-bipartite complement.
For any arc ordering $\prec$ of~$\calN[G]$,
every vertex $u\in V(G)$ has the following property:
$u$ divides $N[u]=[u^-,u^+]$ into two parts $[u^-,u]$~and~$[u,u^+]$
that both are cliques in~$G$.
\end{lemma}

\begin{proof}
    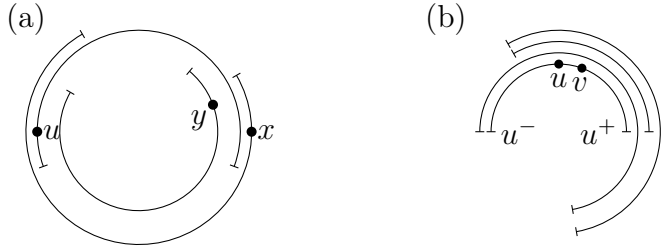
\begin{figure} \centering
      \begin{tikzpicture}[baseline=0cm]
  \node at (-1.5cm,1.5cm) {(a)};
    \begin{camodel}[cabase=.75cm,castep=.15cm,capoints=true]
          \carc[swap,nodeangle=20]{150}{410}{2}{$y$}
          \carc[swap,nodeangle=180]{-20}{200}{4}{$u$}
          \carc[nodeangle=0]{120}{390}{5}{$x$}
        \end{camodel}
      \end{tikzpicture}
\qquad\qquad
  \begin{tikzpicture}[baseline=0cm]
  \node at (-1.5cm,1.5cm) {(b)};
    \begin{camodel}[cabase=.75cm,castep=.15cm]
      \carc[startlabel=$u^-$,endlabel=$u^+$,startlabelpos=inside,endlabelpos=inside]{180}{0}{1}{}
      \carc[swap,nodeangle=90,every label/.style={inner sep=5pt}]{90}{90}{1}{$u$}
       \carc[swap,every label/.style={inner sep=5pt,anchor=80}]{70}{70}{1}{$v$}
      \carc{180}{-80}{2}{}
      \carc{120}{0}{3}{}
      \carc{120}{-80}{4}{}
    \end{camodel}
  \end{tikzpicture}
  \caption{(a) Proof of Lemma~\protect\ref{lem:twocliques}: Vertices
        $u\in V(G)$, $x\in D^+[u]$, and $y\in D^+[u]\cap[u,x]$ along
        with their neighborhoods in $(V(G),\prec)$.
\quad
(b) Proof of Lemma~\protect\ref{lem:CA-order}.3:
Possible mutual positions of~$N[u]$ and~$N[v]$. The most
    inward arc $[u^-,u^+]$ represents~$N[u]$; the other four arcs show
    possible positions of $[v^-,v^+]=N[v]$.}
\label{fig:geomistight+proof}
    \end{figure}
 Call a vertex $x\in N[u]$ a
    \emph{close neighbor} of~$u$ if
    \[
    x\in [u^-,u]\text{ and }u\in [x,x^+]\text{ or if }x\in [u,u^+]\text{ and
    }u\in [x^-,x]
    \]
    and a \emph{distant neighbor} otherwise. 
    Denote the sets of close neighbors of~$u$ in~$[u^-,u]$
    and~$[u,u^+]$ by~$C^-[u]$ and~$C^+[u]$,
    respectively. 
Similarly,  the sets of distant neighbors of~$u$ in~$[u^-,u]$
    and~$[u,u^+]$ will be denoted by~$D^-[u]$ and~$D^+[u]$. 
Each of the four sets~$C^-[u]$, $C^+[u]$, $D^-[u]$,
    and~$D^+[u]$ is a clique.  Indeed, if for example $x$ and~$y$ are
    two vertices in~$D^+(u)$, then $u$ belongs to both $[x,x^+]$
    and~$[y,y^+]$. Since $x$ and $y$ are both in~$[u,u^+]$, this
    implies that either $x$ belongs to~$[y,y^+]$ or $y$ belongs
    to~$[x,x^+]$ (depending on whether $y\in[u,x]$ or
    $x\in[u,y]$; see Fig.~\ref{fig:geomistight+proof}). 
Hence, $x$ and $y$ are adjacent and~$D^+(u)$ is a
    clique. The other three cases are similar.  

    To complete the proof we show that,
if $\barG$ is non-bipartite, then
    $D^-[u]=D^+[u]=\emptyset$.  Assume to the contrary that $D^+[u]$
    contains a vertex~$x$; see Fig.~\ref{fig:geomistight+proof} (the case $x\in
    D^-[u]$ is similar).  Since $u\in [x,x^+]$, the sets
    $[u,u^+]\cap[u,x]=[u,x]$ and $[x,x^+]\cap[x,u]=[x,u]$ cover $V(G)$. Splitting
    the former set into $C^+[u]\cap[u,x]$ and $D^+[u]\cap[u,x]$ and
    the latter into $C^+[x]\cap[x,u]$ and $D^+[x]\cap[x,u]$, consider
    the cover of~$V(G)$ by two sets
    $(C^+[u]\cap[u,x])\cup(D^+[x]\cap[x,u])$ and
    $(D^+[u]\cap[u,x])\cup(C^+[x]\cap[x,u])$ and show that they are
    cliques.  This will give us a contradiction since $\barG$~is not
    bipartite.  By symmetry, it suffices to prove that
    $(D^+[u]\cap[u,x])\cup(C^+[x]\cap[x,u])$ is a clique. Since both
    $D^+[u]$ and~$C^+[x]$ are cliques, we have to show that any
    vertex~$y$ in $D^+[u]\cap[u,x]$ is adjacent to all vertices in
    $C^+[x]\cap[x,u]$.  This is true because we have $
    N[y]\supseteq[y,u]\supseteq[x,u] $ by the definition of~$D^+[u]$.
\end{proof}

Lemma~\ref{lem:twocliques} allows us to derive the following lemma, which will
be needed for the main rigidity result of this section and also is of
independent interest.

\begin{lemma}\label{lem:anyistight}
If $G$ is a PCA graph with non-bipartite complement, 
then any arc ordering $\prec$ of~$\calN[G]$ is tight.
\end{lemma}

This complements Lemma~\ref{lem:geomistight}, which implies that there are tight
arc orderings of $\calN[G]$.

\begin{proof}
Let $N[u]\subseteq N[v]$.
Consider the arcs $N[u]=[u^-,u^+]$ and $N[v]=[v^-,v^+]$ w.r.t.~$\prec$.
Suppose that $u\in [v,v^+]$ (the case that $u\in [v^-,v]$ is symmetric).
By Lemma~\ref{lem:twocliques}, $[v,v^+]$ is a clique in $G$.
Therefore, this arc is contained in~$N[u]$, which implies
that $u^+=v^+$. 
\end{proof}

With the next lemma we show that all arc orderings for closed neighborhood
hypergraph of non-co-bipartite PCA graphs follow several simple patterns.

\begin{lemma}\label{lem:CA-order}
Let $G$ be a PCA graph with non-bipartite complement.
For any arc ordering $\prec$ of~$\calN[G]$ and any vertices $u,v\in V(G)$,
the following conditions are met.
  \begin{bfenumerate}
  \item $v\in [u,u^+]$ if and only if $u\in [v^-,v]$.
  \item If $v\in [u,u^+]$, then $v^-\in [u^-,u]$ and $u^+\in [v,v^+]$.
  \item If $u\prec v$ and these vertices are adjacent, then $u^-$, $v^-$,
    $u$, $v$, $u^+$, and $v^+$ occur under the
    order~$\prec$ exactly in this circular sequence, where some
    of the neighboring vertices except $u^-$~and~$v^+$ may coincide.
  Moreover, it is impossible that $v^+\prec u^-$.
  \end{bfenumerate}
\end{lemma}

\begin{proof}
 {\bf 1.} Let $u\ne v$ and assume that $v\in [u,u^+]$. If $u\in
    [v,v^+]$, then Lemma~\ref{lem:twocliques} shows that 
$V(G)$ is covered by two cliques
    $[u,u^+]$~and~$[v,v^+]$, contradicting the assumption
    that~$\barG$~is not bipartite. Therefore, $u\in [v^-,v]$. The
    other implication follows by symmetry.
  
{\bf 2.} If the two conditions in~part~1 are true, then $v^-$,
    $u$, and $v$ occur in this circular order. Since $[v^-,v]$ is a
    clique, all vertices in $[v^-,u)$ are adjacent to~$u$ and hence,
    $v^-\in [u^-,u]$. The second containment follows
    by symmetry.
  
{\bf 3.} Parts~1~and~2 imply that $u^-$, $v^-$, $u$, $v$,
    $u^+$, and $v^+$ occur in this circular order; see Fig.~\ref{fig:geomistight+proof}.
However, it is still not excluded $v^+$~and~$u^-$  can coincide or be swapped.
We have to show that the condition $u\prec v$ rules out the last two possibilities
as well as the possibility of $v^+\prec u^-$. Indeed, any of these configurations
would give covering of $V(G)$ by two cliques $[u^-,u]$ and~$[v,v^+]$. 
\end{proof}

The following theorem allows us to invoke Theorem~\ref{thm:unique}.2 for the
closed neighborhood hypergraphs of connected non-co-bipartite PCA graphs,
proving that these hypergraphs have a unique tight arc ordering.

\begin{theorem}\label{thm:strict-conn}
  If $G$~is a connected PCA graph with non-bipartite complement,
then $\calN[G]$~is strictly connected.
\end{theorem}

\begin{proof}
  Let $\alpha$ be a proper arc representation of~$G$.
By Lemma~\ref{lem:geomistight}, $\calN[G]$~has an arc 
(geometric) ordering $\prec_\alpha$.
Since the complement of $G$ is not bipartite,
$G$ has at least two vertices and, by Lemma~\ref{lem:nouniv}, no universal vertex.
Since $G$~is connected, there is at most one pair of non-adjacent vertices~$x$ and~$y$
satisfying the relation $x\prec_\alpha y$. Therefore, all vertices of~$G$
can be arranged into a path $v_1,\ldots,v_n$ such that $v_i$ and~$v_{i+1}$
are adjacent and $v_i\prec_\alpha v_{i+1}$ for every $1\le i<n$.
By Lemma~\ref{lem:CA-order}.3, we have $N[v_i]\tied^* N[v_{i+1}]$,
which gives us a strictly connected path passing through all hyperedges of~$\calN[G]$.
\end{proof}

Now we are ready to prove our rigidity result for the closed neighborhood
hypergraphs of PCA graphs.

\begin{theorem}\label{thm:Nrigid}
Let $G$ be a twin-free, connected PCA graph.
\begin{bfenumerate}
\item 
If $\barG$ is non-bipartite, then $\calN[G]$ has, up to reversal,
a unique arc ordering.
\item 
If $\barG$ is bipartite and connected, then $\calN[G]$ has, up to reversal,
exactly two tight arc orderings.
\end{bfenumerate}
\end{theorem}

\begin{proof}
  {\bf 1.}
$\calN[G]$ has an arc ordering by Lemma~\ref{lem:geomistight}.
By Lemma~\ref{lem:anyistight}, any arc ordering of $\calN[G]$
is tight. The uniqueness follows from Theorem~\ref{thm:strict-conn}
by Theorem~\ref{thm:unique}.2.

{\bf 2.} 
The \emph{open neighborhood hypergraph} of a graph $G$ is defined by
$\calN(G)=\{N(v)\}_{v\in V(G)}$.
In place of $\calN[G]$, it is now practical to consider
the complement hypergraph $\overline{\calN[G]}=\Set{V(G)\setminus N[v]}_{v\in V(G)}$.
Note that $\overline{\calN[G]}=\calN(\barG)$.
This hypergraph is disconnected, and any (tight) arc ordering
of $\calN[G]$ induces a (tight) interval ordering of
each connected component of $\calN(\barG)$.
Conversely, arbitrary (tight) interval orderings
of the components of $\calN(\barG)$ can be merged
into a (tight) arc ordering of $\calN[G]$.
Since $\barG$ is connected, $\calN(\barG)$ has exactly two components.
Applying Theorem~\ref{thm:unique}.1 to each of them, we 
conclude that each of the two components has, up to reversing,
a single interval ordering. Since there are exactly two essentially
different ways to merge them, we see 
that $\calN[G]$ has, up to reversing, exactly two tight arc orderings.
\end{proof}

Let us stress that part 2 of Theorem~\ref{thm:Nrigid} concerns 
only tight orderings. To show that it
cannot be strengthened to the class of all arc orderings,
consider the \emph{half-graph}~$H_m$ that is the bipartite graph
with vertex classes $\{u_1,\ldots,u_m\}$ and $\{v_1,\ldots,v_m\}$
where $u_i$ is adjacent to $v_j$ if $i\le j$.
The complement $G_m=\overline{H_m}$ is a twin-free, connected PCA graph.
Note that, besides two pairs of mutually reversed tight interval orderings,
the hypergraph $\calN(H_3)$ has another interval ordering
and, hence, $\calN[G_3]$ has a non-tight arc ordering.
If we increase the parameter $m$, the number
of non-tight arc orderings of $\calN[G_m]$ grows exponentially.

\section{Overlap-connectedness}\label{s:ov-conn}

Let $G$ be a twin-free, connected PCA graph with non-bipartite complement.
In Theorem~\ref{thm:strict-conn} we established that $\calN[G]$
is strictly connected. This implies the uniqueness of a tight arc
ordering of this hypergraph by Theorem~\ref{thm:unique}.2.
Since by Lemma~\ref{lem:anyistight} any arc ordering is actually tight,
we obtain the uniqueness for the class of all arc orderings.

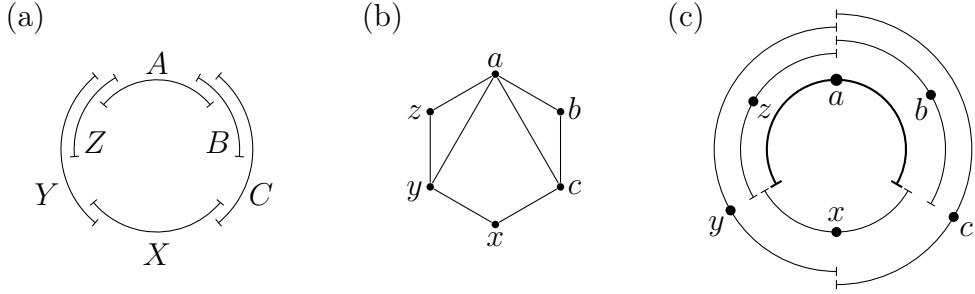
\begin{figure}
  \centering
  \begin{tikzpicture}[baseline=0cm]
    \node at (-1.75cm,1.75cm) {(a)};
    \begin{camodel}[cabase=.75cm,castep=.175cm]
      \carc{140}{40}{1}{$A$}
      \carc{-40}{-140}{2}{\strut$X$}
      \carc[labelpos=.67]{50}{-50}{3}{$\!C$}
      \carc[labelpos=.33]{230}{130}{3}{$Y$}
      \carc[swap,labelpos=.67]{60}{-5}{2}{$B$}
      \carc[swap,labelpos=.33]{185}{120}{2}{$Z$}
    \end{camodel}
  \end{tikzpicture}\hfil
  \begin{tikzpicture}[baseline=0cm]
    \node at (-1.5cm,1.75cm) {(b)};
    \path[every node/.style={circle,fill,inner sep=1pt},every label/.append style=rectangle]
        (90:1cm) node[label=above:$a$] (a)   {}
        (30:1cm) node[label=right:\strut$b$] (b)   {} edge (a)
        (-30:1cm) node[label=right:$c$] (c) {} edge (a) edge (b)
        (-90:1cm) node[label=below:$x$] (x) {} edge (c)
        (-150:1cm) node[label=left:\strut$y$] (y) {} edge (a) edge (x)
        (150:1cm) node[label=left:$z$] (z) {} edge (a) edge (y);
  \end{tikzpicture}\hfil
  \begin{tikzpicture}[baseline=0cm]
    \node at (-2cm,1.75cm) {(c)};
    \begin{camodel}[capoints=true,cabase=.75cm,castep=.175cm,
      spaced/.style={every label/.append style={inner sep=4pt}},
      close/.style={every label/.append style={inner sep=1pt}}]
      \carc[nodeangle=90,thick,swap,spaced]{210}{-30}{1}{$a$}
      \carc[nodeangle=-90,swap,spaced]{-30}{-150}{2}{$x$}
      \carc[nodeangle=150,swap,close]{210}{90}{3}{$z$}
      \carc[nodeangle=30,swap,close]{90}{-30}{4}{$b$}
      \carc[nodeangle=-150]{-90}{-270}{5}{$y$}
      \carc[nodeangle=-30]{90}{-90}{6}{$c$}
    \end{camodel}
  \end{tikzpicture}%
  \caption{An example: (a) A proper arc system; (b) The corresponding intersection graph~$G$.
Its complement $\barG$~is non-bipartite; (c) The closed neighborhood hypergraph~$\calN[G]$ is not strictly
overlap-connected: the hyperedge~$N[a]$ forms a single component.
Nevertheless, since after removal of~$N[a]$ the hypergraph~$\calN[G]$ stays twin-free
and becomes strictly overlap-connected, it has a unique, up to reversal, arc ordering.}\label{fig:example}
\end{figure}

The same conclusion could be derived more directly
by Theorem~\ref{thm:unique-overlap-2} when $\calN[G]$
would be strictly overlap-connected.
However, the last condition is not always true
because $\calN[G]$ can have hyperedges of size $n-1$ and
each such hyperedge forms a separate
strictly overlap-connected component; see an example in Fig.~\ref{fig:example}.
Nevertheless, if we remove the $(n-1)$-element hyperedge
from $\calN[G]$ in this example, the remaining hypergraph becomes
strictly overlap-connected and stays twin-free.
It turns out that this is a general phenomenon.
In fact, we derive this result from the uniqueness of an arc ordering, 
that we already established for $\calN[G]$ in the preceding section,
and a criterion of uniqueness given below.

\begin{lemma}\label{lem:ov-conn-gen}
Given a hypergraph~$\calH$ on $n\ge4$ vertices, let $\calH'$ 
be the hypergraph on the same vertex set 
obtained from~$\calH$ by removing all hyperedges of size $0$, $1$, $n-1$, and~$n$.
Then $\calH$~has a unique, up to reversing, arc ordering if and only if 
$\calH'$~is twin-free and either is strictly overlap-connected
or has a single isolated vertex and becomes strictly overlap-connected
after its removal.
\end{lemma}

\begin{proof}
Assume that $\calH'$~is twin-free and either is strictly overlap-connected
or becomes so after removing a single isolated vertex.
By Theorem~\ref{thm:unique-overlap-2}, $\calH'$~has a unique arc ordering. 
Recall that uniqueness is always meant up to reversal.
This holds also for~$\calH$ because the two hypergraphs have
the same arc orderings.

Let us prove the lemma in the other direction.
Assume that $\calH$~has a unique arc ordering.
Since $n\ge4$, $\calH$~has no twins (for else transposition of two twins would give
us another arc ordering).

Fix an arbitrary vertex~$x$ of~$\calH$. For a hyperedge $H\in\calH$,
set $H_x=H$ if $x\notin H$ and $H_x=V(\calH)\setminus H$ otherwise.
Define the interval hypergraph $\calH_x=\Set{H_x}_{H\in\calH}$ 
where any empty hyperedge $H_x=\emptyset$ is removed.
If vertices~$u$ and~$v$ are distinguished by the incidence relation to a hyperedge~$H$, 
they are distinguished as well by the complement of $H$. This shows that $\calH_x$~is, like $\calH$, twin-free.
In particular, $x$~is the only isolated vertex of~$\calH_x$.
Let $\calH_x^\circ$ denote the hypergraph obtained from~$\calH_x$ by removing the vertex~$x$.

Any arc ordering for $\calH_x$~is also an arc ordering for~$\calH$.
Therefore $\calH_x$~has a unique arc ordering.
It readily follows that $\calH_x^\circ$~has a unique interval ordering.
The set of all interval orderings of a given hypergraph is described in~\cite{KoeblerKLV11}
in terms of the tree of its overlap-connected components, which is an analog of
the classical structure known as a \emph{PQ-tree}~\cite{BoothL76}. 
This description is based on an observation that any two overlap-connected components
of a hypergraph either are vertex-disjoint or one is contained in a twin-class of the other.
With respect to this containment relation, the overlap-connected components
form a directed forest. By Theorem~\ref{thm:unique-overlap-1}, every overlap-connected component
admits a unique interval ordering, and we have a freedom to reverse it within each of the
components. It follows that $\calH_x^\circ$~is connected 
and that the root overlap-connected component $\calR$ of~$\calH_x^\circ$ is twin-free.
Therefore, any other overlap-connected component of~$\calH_x^\circ$
must consist of a single one-vertex hyperedge (possibly obtained by complementing
a $(n-1)$-vertex hyperedge of~$\calH$).
It follows that $\calR$
is equal to the hypergraph~$(\calH')_x^\circ$ obtained from~$\calH'$ by complementing
all hyperedges that contain~$x$ and removing~$x$.
Thus, $(\calH')_x^\circ$ is twin-free and overlap-connected.
Since $(\calH')_x$~is obtainable from the connected hypergraph~$(\calH')_x^\circ$ by adding
an isolated vertex, $(\calH')_x$~is twin-free as well.
By the same argument as used above for~$\calH$ and~$\calH_x$, the hypergraph~$\calH'$ must be twin-free too.
In particular, $\calH'$~has at most one isolated vertex.

Note that if hyperedges~$H_x$ and~$K_x$ of the hypergraph~$(\calH')_x$
overlap, then they strictly overlap in the full circle $V(\calH)$ and, therefore,
the corresponding hyperedges~$H$ and~$K$ of the hypergraph~$\calH'$
strictly overlap. It now follows from the overlap-connectedness of~$(\calH')_x^\circ$ that
any two hyperedges in~$\calH'$ are connected by a $\between^*$-path.
This readily implies that $\calH'$~is strictly overlap-connected if it has no isolated
vertex or becomes such after removal of the (single) isolated vertex.
\end{proof}

Recall that $\calK$ is a \emph{spanning subhypergraph} of
a hypergraph~$\calH$ if
$V(\calK)=V(\calH)$ and any hyperedge of~$\calK$ is a hyperedge of~$\calH$.

\begin{theorem}\label{thm:ov-conn-str}
If $G$~is a twin-free, non-co-bipartite, connected PCA graph on~$n$ vertices, then
 the spanning subhypergraph of~$\calN[G]$ obtained by removing
all hyperedges of size $n-1$ is twin-free and strictly overlap-connected.
\end{theorem}

\begin{proof}
Assume that a graph~$G$ satisfies all the assumptions.
Note that then $G$~has $n\ge4$ vertices.
Denote $\calH=\calN[G]$. 
Since $G$~is connected, $\calH$~has no hyperedge of size~$1$. 
It also has no hyperedge of size $n$
because $G$~has no universal vertex by Lemma~\ref{lem:nouniv}.
Remove all hyperedges of size $n-1$ and denote the result by~$\calH'$.
We have to show that $\calH'$~is twin-free and strictly overlap-connected.

By Theorem~\ref{thm:Nrigid}.1,
$\calH$~has a unique, up to reversing, arc ordering.
By Lemma~\ref{lem:ov-conn-gen},
$\calH'$~is twin-free and strictly overlap-connected unless it has an isolated vertex.
It remains to exclude the last possibility.

Suppose that all hyperedges of~$\calH$ containing a vertex~$v$ are of size $n-1$.
In particular, $N[v]=V(G)\setminus\{v'\}$, where $v'$~is a unique vertex non-adjacent
to~$v$. Note that $v$~is contained in the hyperedge~$N[u]$ for any $u\ne v'$.
For each of these hyperedges we, therefore, have $|N[u]|=n-1$.
Moreover, $v'\in N[u]$ for else $u$ and~$v$ would be twins.
It follows that $|N[v']|=n-1$ as well. We conclude that $\barG$~is a perfect matching,
contradicting the assumptions that $\barG$~is not bipartite.  
\end{proof}

Note that Theorem~\ref{thm:ov-conn-str} gives us a complete description of the decomposition of~$\calN[G]$
into strictly overlap-connected components because each hyperedge  of size $n-1$ forms such a 
component alone. Figure~\ref{fig:manylargehyperedges} shows that the number of
such single-hyperedge components can be linear.

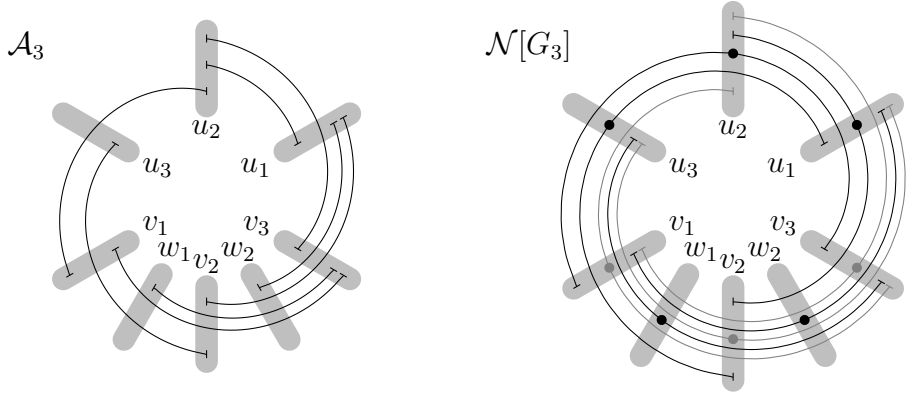
\begin{figure}
  \centering
  \numdef\k{3}
  \numdef\n{3*\k-1}
  \numdef\p{2*\k}
  \numdef\km{\k-1}
  \numdef\rot{-180/\p}
  \dimdef\base{1.2cm}
  \begin{tikzpicture}
    \begin{camodel}[cabase=\base,caanglestep=360/\p,castep=.1cm]
      \node[anchor=east] at (-2,2) {$\calA_3$};
      \begin{scope}[rotate=\rot]
        \begin{pgfonlayer}{background}
          \draw[black!25,line width=3mm,line cap=round]
          (360/\p:1cm+\base) -- (360/\p:\base)
          node[black,anchor=\rot+360/\p,inner sep=1pt] {$u_1$};
        \end{pgfonlayer}
        \CArc{start=1,end=(\k-1),startlevel=2,endlevel=5.5}
        \foreach \i in {2,...,\k} {%
          \begin{pgfonlayer}{background}
            \draw[black!25,line width=3mm,line cap=round]
            (\i*360/\p:1cm+\base) -- (\i*360/\p:\base)
            node[black,anchor=\rot+\i*360/\p,inner sep=1pt] {$u_\i$};
          \end{pgfonlayer}
          \CArc{start=\i,end=(\k+\i-1),startlevel=2,endlevel=(3*\k)}
        }
        \foreach \i in {1,...,\k} {%
          \begin{pgfonlayer}{background}
            \draw[black!25,line width=3mm,line cap=round]
            ({(\k+\i)*360/\p}:1cm+\base) -- ({(\k+\i)*360/\p}:\base)
            node[black,anchor=\rot+(\k+\i)*360/\p,inner sep=1pt] {$v_\i$};
          \end{pgfonlayer}
          \CArc{start=(\k+\i),end=(\p+\i-1),startlevel=2,endlevel=(3*\k)}
        }
        \foreach \i in {1,...,\km} {%
          \begin{pgfonlayer}{background}
            \draw[black!25,line width=3mm,line cap=round]
            ({(\k+\i+1/2)*360/\p}:1cm+\base) -- ({(\k+\i+1/2)*360/\p}:\base)
            node[black,anchor=\rot+(\k+\i+1/2)*360/\p,inner sep=1pt] {$w_\i$};
          \end{pgfonlayer}
          \CArc{start=(\k+\i+1/2),end=(\p+\i-1),startlevel=2,endlevel=(3*\k-1.5)}
        }
      \end{scope}
    \end{camodel}
    \begin{camodel}[cabase=\base,xshift=7cm,caanglestep=360/\p,castep=.1cm,capoints=true]
      \node[anchor=east] at (-2,2) {$\calN[G_3]$};
      \begin{scope}[rotate=\rot]
        \begin{pgfonlayer}{background}
          \draw[black!25,line width=3mm,line cap=round]
          (360/\p:1.25cm+\base) -- (360/\p:\base)
          node[black,anchor=\rot+360/\p,inner sep=1pt] {$u_1$};
        \end{pgfonlayer}
        \CArc{start=(1+1-\k),end=(\k-1),startlevel=2,endlevel=(4*\k-2.5),labelangle=1}
        \foreach \i in {2,...,\k} {%
          \begin{pgfonlayer}{background}
            \draw[black!25,line width=3mm,line cap=round]
            (\i*360/\p:1.25cm+\base) -- (\i*360/\p:\base)
            node[black,anchor=\rot+\i*360/\p,inner sep=1pt] {$u_\i$};
          \end{pgfonlayer}
          \CArc{start=(\i+1-\k),end=(\k+\i-1),startlevel=2,endlevel=(4*\k)} }
        \foreach \i in {1,...,\k} {%
          \begin{pgfonlayer}{background}
            \draw[black!25,line width=3mm,line cap=round]
            ({(\k+\i)*360/\p}:1.25cm+\base) -- ({(\k+\i)*360/\p}:\base)
            node[black,anchor=\rot+(\k+\i)*360/\p,inner sep=1pt] {$v_\i$};
          \end{pgfonlayer}
          \CArc{gray,start=(\i+1),end=(\p+\i-1),startlevel=2,endlevel=(4*\k)} }
        \foreach \i in {1,...,\km} {%
          \begin{pgfonlayer}{background}
            \draw[black!25,line width=3mm,line cap=round]
            ({(\k+\i+1/2)*360/\p}:1.25cm+\base) -- ({(\k+\i+1/2)*360/\p}:\base)
            node[black,anchor=\rot+(\k+\i+1/2)*360/\p,inner sep=1pt] {$w_\i$};
          \end{pgfonlayer}
          \CArc{start=(\i+2),end=(\p+\i-1),startlevel=3.25,endlevel=(4*\k-1.25)}
        }
      \end{scope}
    \end{camodel}
  \end{tikzpicture}
  \caption{A non-co-bipartite, twin-free, connected PCA graph on~$n$ vertices
    can have more than $n/3$ vertices of degree $n-2$ (a vertex~$v$ has degree $n-2$ iff
    $\left|N[v]\right|=n-1$). For each $k\ge2$, define a graph~$G_k$ on $3k-1$
    vertices by its arc model~$\calA_k$ such that $k$ of the vertices
will have degree $n-2$. On the circle
    $\circl=\{u_1,u_2,\dotsc,u_k,v_1,w_1,v_2,w_2,\dotsc,v_{k-1},w_{k-1},v_k\}$,
    whose points go in this circular order, consider arcs
    $U_1=[u_1,u_{k-1}]$, $U_i=[u_i,v_{i-1}]$ for $2\le i\le k$, 
    $V_1=[v_1,v_k]$, $V_i=[v_i,u_{i-1}]=\circl\setminus U_i$ for $2\le i\le k$, and
    $W_i=V_i\setminus\{v_i\}$ for $i\le k-1$. As $\calA_k$~is tight, it can easily
be made proper and, hence, $G_k$~is PCA. Since the arcs have pairwise
different left endpoints, we can identify $V(G_k)=\circl$. The graph can be described
by listing the pairs of non-adjacent vertices, namely $v_iu_i$, $w_iu_i$, $w_iu_{i+1}$,
and~$u_1u_k$. There are no twins, and $G_k$~is not co-bipartite
because its complement contains an odd cycle, namely $u_1w_1u_2w_2\ldots u_k$.
In accordance with Theorem~\protect\ref{thm:ov-conn-str}, 
when we remove from~$\calN[G_k]$ the hyperedges $N[v_i]=\circl\setminus\{u_i\}$,
the hypergraph stays twin-free and becomes strictly overlap-connected.
This can be seen by looking at the following path in the complementing hypergraph:
$\overline{N[u_1]}=\{u_k,v_1,w_1\}$, $\overline{N[u_2]}=\{w_1,v_2,w_2\}$, \ldots,
$\overline{N[u_{k-1}]}=\{w_{k-2},v_{k-1},w_{k-1}\}$, $\overline{N[u_k]}=\{w_{k-1},v_k,u_1\}$,
$\overline{N[w_1]}=\{u_1,u_2\}$, \ldots, $\overline{N[w_{k-1}]}=\{u_{k-1},u_k\}$.
The figure shows~$\calA_3$
    and an arc model for~$\calN[G_3]$ which has the arcs of size $n-1$ grayed
    out.}\label{fig:manylargehyperedges}
\end{figure}

\section{Intersection representations of graphs}\label{s:repr}

Part 1 of Theorem~\ref{thm:Nrigid} implies that a twin-free, connected PCA graph $G$
with non-bipartite complement has a unique geometric order.
If $\barG$ is bipartite and connected, part~2 (combined with Lemma~\ref{lem:geomistight})
leaves two different possibilities. It turns out that, nevertheless, 
a geometric order is unique also in this case. We state this fact
in Theorem~\ref{thm:Aunique2} below, proving an auxiliary lemma
beforehand.

\begin{lemma}\label{lem:one-endpoint}
If a graph $G$ contains at most one universal vertex, then
any proper arc representation $\alpha$ of $G$ has the following property:
If $v$ and $v'$ are adjacent vertices of $G$,
then the arcs~$\alpha(v)$ and~$\alpha(v')$
strictly overlap, that is, contain exactly one endpoint of each other.  
\end{lemma}

\begin{proof}
Otherwise we would have either $\alpha(v)\subset\alpha(v')$,
or $\alpha(v)\supset\alpha(v')$, or the union $\alpha(v)\cup\alpha(v')$
would cover the whole circle.
The first two conditions would contradict the assumption that $\alpha$~is proper,
while the last condition would imply that every $\alpha(u)$ intersects both $\alpha(v)$
and~$\alpha(v')$. Therefore, both $v$ and $v'$ were universal,
a contradiction.
\end{proof}

\begin{theorem}\label{thm:Aunique2}
Let $G$ be a twin-free, connected PCA graph.
If its complement $\barG$ is non-bipartite or connected,
then all geometric orders associated with proper arc representations of~$G$
are equal up to reversing.
\end{theorem}

\begin{proof}
As we just discussed, it is enough to consider the case that $\barG$
is bipartite and connected.
  Let $V(G)=U\cup W$ be the bipartition of~$\barG$ into two independent sets.
Note that $U$ and $W$ span two connected components of the hypergraph $\calN(\barG)$.
Consider an arbitrary proper arc representation~$\alpha$ of~$G$ and the corresponding
geometric order $\prec_\alpha$ on~$V(G)$. 
Like in the proof of Theorem~\ref{thm:Nrigid}.2,
notice that $\prec_\alpha$~is a tight
interval order for each connected component of~$\calN(\barG)$ and, therefore,
the restrictions of~$\prec_\alpha$ to~$U$ and~$W$,
each considered up to reversing, do not depend on~$\alpha$ by Theorem~\ref{thm:unique}.1.
This still leaves two distinct possibilities of merging them into~$\prec_\alpha$.
We now show that one of them is actually ruled out, and this will
prove that $\prec_\alpha$ does not depend on~$\alpha$, 
if this order is considered up to reversing.

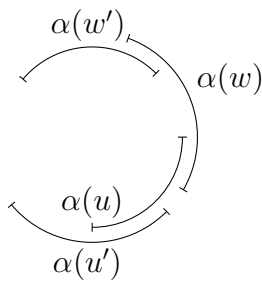
\begin{figure}
  \centering
\begin{tikzpicture}
 \begin{camodel}
  \carc{140}{45}{1}{$\alpha(w')$}
  \carc[endlabel=$\alpha(u)$,endlabelpos=inside]{0}{-90}{1}{}
  \carc{70}{-30}{2}{$\alpha(w)$}
  \carc{-45}{-140}{2}{$\alpha(u')$}
 \end{camodel}
\end{tikzpicture}
  \caption{Proof of Theorem~\protect\ref{thm:Aunique2}.}\label{fig:wwuu}
\end{figure}

Note that, since $\barG$ is connected, $G$ cannot have universal vertices
and we can use Lemma~\ref{lem:one-endpoint}.

The assumptions of the theorem imply that both $U$ and~$W$ contain at least two
vertices. Let $w$ and~$w'$ be two vertices in~$W$. Since $w$ and~$w'$ are not twins,
they are distinguished by adjacency to some vertex $u\in U$. W.l.o.g.\
suppose that $w$ and~$u$ are adjacent, while $w'$ and~$u$ are not.
Since $w$~is not universal in~$G$, there is a vertex $u'$
in~$U$ non-adjacent to~$w$. Note that $\alpha(w')$ and~$\alpha(u)$
contain different endpoints of~$\alpha(w)$, and $\alpha(u')$ and~$\alpha(w)$
contain different endpoints of~$\alpha(u)$; see Fig.~\ref{fig:wwuu}. 
It follows that the arcs~$\alpha(w')$, $\alpha(w)$, $\alpha(u)$, and~$\alpha(u')$ occur
in the arc model exactly in this circular order, irrespective of whether $\alpha(w')$
and~$\alpha(u')$ intersect or not. 
Since the quadruple $w',w,u,u'$ was chosen in terms of the graph~$G$ alone
and $\alpha$ was supposed to be arbitrary, this conclusion holds true for
any proper arc representation of~$G$. Therefore, there is a unique way
of merging the restrictions of~$\prec_\alpha$ to~$U$ and~$W$ so as to obtain
an ordering of~$V(G)$ consistent with some geometric order.
This proves the desired uniqueness result.
\end{proof}

If a graph has one interval or arc representation, 
we can obtain many other representations, for example, by cloning some points
of the circle. To disallow such trivial modifications, let us
impose some nonrestrictive conditions on interval and arc models.
Let $x$ be a point in an arc system~$\calA$. We call $x$ \emph{inner}
if no arc of~$\calA$ begins or ends at~$x$. 
Suppose now that $\calA$~is a proper intersection model of a graph~$G$.
We can modify $\calA$ so that it remains a proper model of~$G$
while the following two conditions are true:
\begin{itemize}
\item 
no two arcs share an endpoint point; this can always be achieved by cloning
a shared endpoint point (if $x$ the right endpoint of~$A_1$ and the left endpoint of~$A_2$,
replace $x$ with the pair $x_1,x_2$ so that $A_1$ ends at $x_1$, $A_2$ starts at $x_2$,
and $x_1$~is the right neighbor of~$x_2$);
\item 
there is no inner point (just remove all of them).
\end{itemize}
We call such arc models and representations \emph{sharp}. 
Sharp interval representations are defined similarly.
Note that, if $G$~has $n$ vertices, any sharp
model of~$G$~has $m=2n$ points.

We now show that a sharp proper representation is
reconstructible from the associated geometric order.
Given two sharp arc representations $\alpha$ and $\alpha'$
of a graph $G$, we say that they are equal \emph{up to rotation}
of the circle $\circl_{m}$ if there is a rotation $\sigma$
of $\circl_{m}$ such that $\alpha'=\sigma\circ\alpha$.
If $\sigma$ is allowed to be also the reflection of $\circl_{m}$,
then we say that $\alpha$ and $\alpha'$ are equal \emph{up to rotation and reflection}.

\begin{lemma}\label{lem:ReprOrder}
\begin{bfenumerate}
\item
If sharp proper interval representations of a graph~$G$ determine the same
geometric order, then they are equal; that is, if $\prec_\alpha=\prec_{\alpha'}$,
then $\alpha=\alpha'$.
\item
Suppose that $G$~has at most one universal vertex.
If sharp proper arc representations of~$G$ determine the same geometric order, 
then they are equal up to rotation.
\end{bfenumerate}
\end{lemma}

\begin{proof}
{\bf 1.}
Let $\alpha$ be a sharp proper interval representation of a graph~$G$ and
$v_1,\ldots,v_n$ be the ordering of~$V(G)$ according to the geometric order $\prec_\alpha$
associated with $\alpha$. It suffices to notice that
$\alpha$~is completely determined by this order. Indeed, let $\alpha(v_i)=[a_i,b_i]$.
Then we must have
\begin{eqnarray}
b_i&=&a_i+1+|N(v_i)|\text{\ \ and}\label{eq:bi}\\      
a_i&=&i+\left|\setdef{j<i}{v_j\notin N[v_i]}\right|.\label{eq:ai}
\end{eqnarray}

{\bf 2.}
Given the geometric order of the vertex set 
associated with a sharp proper arc representation~$\alpha$ of~$G$,
we have to show that it determines $\alpha$ up to rotation.
Fix a sequence $v_1,\ldots,v_n$ according to this order.
From this sequence we can determine the clockwise neighborhood~$N^+[v_i]$
and the counter-clockwise neighborhood~$N^-[v_i]$ of any vertex~$v_i$.
Suppose that $\alpha(v_i)=[a_i,b_i]$ and $a_1=1$. The latter condition
can be ensured by shifting (renaming) the points in the circle~$\circl_{2n}$,
which results just in a rotation of~$\alpha$. 

By Lemma~\ref{lem:one-endpoint}, if
$v_i$~is adjacent to~$v_j$ then $\alpha(v_i)$ contains exactly
one of the endpoints~$a_j$ and~$b_j$. 
We can now see that the start points~$a_i$ 
and end points~$b_i$ are uniquely determined. Equality~\refeq{bi} holds true
exactly as in part~1; whenever the right hand side exceeds $2n$, it has
to be decreased by this number. Equality~\refeq{ai} has to be adjusted. 
If $v_i\notin N^+[v_1]$, then
\[
[1,a_i]=[1,b_1]\cup\setdef{a_j}{1<j\le i,a_j\notin[1,b_1]}\cup
\setdef{b_j}{1<j<i,b_j\notin[a_i,b_i]}.
\]
It follows that
\[
a_i=|[1,a_i]|=2+|N(v_1)|+\left|\setdef{j}{1<j\le i,v_j\notin N^+[v_1]}\right|+
\left|\setdef{j}{1<j<i,v_j\notin N^-[v_i]}\right|.
\]
If $v_i\in N^+[v_1]$, then
\[
[1,a_i]=\setdef{a_j}{1\le j\le i}\cup\setdef{b_j}{v_j\in N^-[v_1]\setminus N^-[v_i]},
\]
hence
\[
a_i=|[1,a_i]|=i+|N^-[v_1]\setminus N^-[v_i]|,
\]
completing the proof.
\end{proof}

Combining Lemma~\ref{lem:ReprOrder} with Theorem~\ref{thm:Aunique2}
(or with Roberts' uniqueness theorem for proper interval graphs),
we immediately obtain the following result.

\begin{theorem}\label{thm:prop-repr-unique}\mbox{}
  \begin{bfenumerate}
  \item 
A twin-free, connected proper interval graph has, up to reflection,
a unique sharp proper interval representation.
  \item 
A twin-free, connected PCA graph with non-bipartite or connected complement has, 
up to rotation and reflection, a unique sharp proper arc representation.
  \end{bfenumerate}
\end{theorem}

\paragraph{Straight and round orientations.}
We conclude with discussing how Theorem~\ref{thm:prop-repr-unique} is related
to the results of Deng, Hell, and Huang~\cite{DengHH96}.
We begin with an overview of some concepts introduced in~\cite{DengHH96}.
Recall that an \emph{orientation} of a graph $G$ is a directed graph
obtained from $G$ by directing each edge. A directed graph obtained
in such a way is called \emph{oriented}.
Let $D$ be an orientation of a graph $G$.
A \emph{straight enumeration} of $D$ is an interval ordering of $\calN[G]$
such that every vertex $v$ splits the interval $N[v]=[v^-,v^+]$
into the set $[v^-,v)$ of in-neighbors of $v$ and into the set
$(v,v^+]$ of out-neighbors of $v$. Similarly, 
a \emph{round enumeration} of $D$ is an arc ordering of $\calN[G]$
such that $[v^-,v)$ and $(v,v^+]$ consist of, respectively, the in- and the out-neighbors
of $v$. If the orientation $D$ admits a straight (resp.\ round) enumeration,
it is called \emph{straight (resp.\ round)}.
Deng, Hell, and Huang~\cite{DengHH96} notice a close connection between
proper interval/arc representations of a graph $G$ and its straight/round orientations:
$G$ is proper interval (resp.\ PCA) if and only if its has a straight (resp.\ round)
orientation. This connection enables obtaining Theorem~\ref{thm:prop-repr-unique}
from the following result in~\cite{DengHH96}.

\begin{theorem}[cf.\ Deng, Hell, and Huang~\cite{DengHH96}]\label{thm:orient-unique}\mbox{}
  \begin{bfenumerate}
  \item 
A twin-free, connected proper interval graph has, up to reversal,
a unique straight orientation.
  \item 
A twin-free, connected PCA graph with non-bipartite or connected complement has, 
up to reversal, a unique round orientation.
  \end{bfenumerate}
\end{theorem}

Though obtained in our paper and in~\cite{DengHH96} by completely different methods, 
Theorems~\ref{thm:prop-repr-unique} and~\ref{thm:orient-unique}
are equivalent statements of the same nature. For completeness we now
derive the latter result from the former.

\begin{proofof}{Theorem~\ref{thm:orient-unique}}
We show part 2; part 1 is similar.
Following~\cite{DengHH96}, we first notice that a proper arc representation of
a twin-free graph $G$ determines a round orientation of $G$ and vice versa.

Given a proper arc representation $\alpha$ of $G$,
define an orientation $D_\alpha$ of $G$ by directing each edge $\{u,v\}$
as $uv$ if $\alpha(u)$ contains the left endpoint of $\alpha(v)$.
The definition is unambiguous by Lemma~\ref{lem:one-endpoint}.
Note that the geometric order $\prec_\alpha$ is a round enumeration of~$D_\alpha$.

Note also that, if $\alpha'$ is a rotation of $\alpha$, then
$D_{\alpha'}=D_\alpha$. If $\alpha'$ is the reflection of $\alpha$, then
$D_{\alpha'}$ is the reversed version of $D_\alpha$.

As observed in~\cite[Proposition 2.6]{DengHH96} (see also~\cite[Theorem 3.1]{HellH95}),
every round orientation $D$ of $G$ is obtainable in this way,
that is, $D=D_\alpha$ for some proper arc representation $\alpha$ of $G$.
Under the assumptions made for $G$, Theorem~\ref{thm:prop-repr-unique}
states that $\alpha$ is unique up to rotation and reflection.
It follows that $D=D_\alpha$ is unique up to reversal.
\end{proofof}

\end{document}
